\documentclass[article]{IEEEtran}

\usepackage[english]{babel}
\usepackage[latin1]{inputenc}
\usepackage{enumerate}
\usepackage{color}
\usepackage[T1]{fontenc}
\usepackage{epstopdf}
\usepackage{subfigure}
\usepackage{dsfont}
\usepackage[final]{graphicx}
\DeclareGraphicsExtensions{.eps}
\usepackage[T1]{fontenc}
\usepackage{amsmath}
\usepackage{mathtools}
\usepackage{amsthm}
\usepackage{amstext}
\usepackage{amssymb}
\usepackage{mathrsfs}
\usepackage{cite}
\usepackage{mathtools}
\usepackage{tikz}
\usetikzlibrary{arrows,positioning}

\def\ess{{\mathrm{ess}}}
\def\rL{\mathcal{L}}

\def\R{\mathbb{R}}
\def\rP{\mathbb{P}}

\def\supp{\mathop{\rm supp}}

\def\Proj{\mathop{\rm Proj}}

\def\B{{\mathcal B}}

\def\P{{\mathcal P}}

\def\M{{\mathcal M}}

\def\bu{{\bf u}}
\def\by{{\bf y}}

\def\sX{{\mathsf X}}
\def\sY{{\mathsf Y}}

\def\sE{{\mathsf E}}

\def\sU{{\mathsf U}}
\def\sg{{\mathsf g}}

\theoremstyle{remark}

\newtheorem{definition}{Definition}
\newtheorem{theorem}{Theorem}

\newtheorem{proposition}{Proposition}
\newtheorem{lemma}{Lemma}
\theoremstyle{remark}
\newtheorem{remark}{Remark}
\newtheorem{assumption}{Assumption}

\IEEEoverridecommandlockouts

\allowdisplaybreaks

\begin{document}
\sloppy
\title{A topology for Team Policies and Existence of Optimal Team Policies in Stochastic Team Theory
\author{Naci Saldi
\thanks{The author is with the Department of Natural and Mathematical Sciences, Ozyegin University, Cekmekoy, Istanbul, Turkey, Email: naci.saldi@ozyegin.edu.tr}
}}
\maketitle

\begin{abstract}
In this paper, we establish the existence of team-optimal policies for static teams and a class of sequential dynamic teams. We first consider the static team problems and show the existence of optimal policies under certain regularity conditions on the observation channels by introducing a topology on the set of policies. Then we consider sequential dynamic teams and establish the existence of an optimal policy via the static reduction method of Witsenhausen. We apply our findings to the well-known counterexample of Witsenhausen and the Gaussian relay channel problem.
\end{abstract}

\section{Introduction}\label{sec1}

Team decision theory has been introduced by Marschak \cite{mar55} to study the behaviour of a group of agents who act collectively in a decentralized fashion in order to optimize a common cost function. Radner \cite{rad62} established fundamental results for static teams, and in particular, demonstrated connections between person-by-person optimality and team-optimality. Witsenhausen's seminal papers \cite{wit71,wit75,wit88,WitsenStandard,WitsenhausenSIAM71,Wit68} on dynamic teams and characterization and classification of information structures have been crucial in the progress of our understanding of dynamic teams. Particularly, well-known counterexample of Witsenhausen \cite{Wit68} demonstrated the challenges that arise due to a decentralized information structure in such models. We refer the reader to \cite{YukselBasarBook} for a more comprehensive overview of team decision theory and a detailed literature review.

The key underlying difference of team decision problems from the classical (centralized) decision problems is the decentralized nature of the information structure; that is, agents cannot share all the information they have with other agents. This decentralized nature of the information structure prevents one to use classical tools in centralized decision theory such as dynamic programming, convex analytic methods, and linear programming. For this reason, establishing the existence and structure of optimal policies is a quite challenging problem in team decision theory.

In this paper, our aim is to study the existence of an optimal policy for team decision problems. In particular, we are interested in sequential team models.

In the literature relatively few results are available on the existence of team optimal solutions. Indeed, so far, existence of optimal policies for static teams and a class of sequential dynamic teams has been studied recently in \cite{GuYuBaLa15,YuSa17}. In these papers, existence of team optimal policies is established via strategic measure approach where strategic measures are the probability measures induced by policies on the product of state space, observation spaces, and action spaces. In this approach, one first identifies a topology on the set of strategic measures and then proves the relative compactness of this set along with the lower semi-continuity of the cost function. If the set of strategic measures are closed,
then one can show the existence of optimal policy via Weierstrass Extreme Value Theorem. However, to establish the closeness of the set of strategic measures, one needs somewhat strong assumptions on the observation channels. For instance, conditions imposed in \cite[Assumption 3.1]{GuYuBaLa15} to establish the closeness of the set of strategic measures with respect to the weak topology implies that observation and their reverse channels are \emph{uniformly} continuous with respect to the total variation distance. The reason for imposing such a strong condition on the observation channel is that convergence with respect to the topology defined on the set of strategic measures does not in general preserve the information structure of the problem (see, e.g., \cite[Theorem 2.7]{YuSa17}).

In this paper, we prove the existence of team optimal policies under the assumption that the observation channels are continuous with respect to the total variation distance and we did not put any restriction on the reverse channels. Unlike strategic measure approach, we introduce a topology on the set of policies, inspired by the topology introduced in \cite[Section 2.4]{BoArGh12}, instead of the set of strategic measures. In this way, we can preserve the information structure of the problem under the convergence of this topology. We first establish the result for static teams. Then, using static reduction of Witsenhausen, we consider the sequential dynamic teams and prove the existence of team optimal solutions using the result in static case. We then apply our findings to counterexample of Witsenhausen and Gaussian relay channel problem.

The rest of the paper is organized as follows. In Section~\ref{sub1sec1} we review the definition of Witsenhausen's \emph{intrinsic model} for sequential team problems. In Section~\ref{sec2} we prove the existence of team optimal solutions for static team problems. In Section~\ref{sec3} we consider the existence of an optimal policy for dynamic team
problems via the static reduction method. In Sections~\ref{witsen} and \ref{gaussrelay} we apply the results derived in Section~\ref{sec3} to study the existence of optimal policies for Witsenhausen's counterexample and the Gaussian relay channel. Section~\ref{conc} concludes the paper.

\subsection{Notation and Conventions}

For a metric space $\sE$, the Borel $\sigma$-algebra (the smallest $\sigma$-algebra that contains the open sets of $\sE$) is denoted by $\cal{E}$.  We let $C_0(\sE)$ and $C_c(\sE)$ denote the set of all continuous real functions on $\sE$ vanishing at infinity and the set of all continuous real functions on $\sE$ with compact support, respectively. For any $g \in C_c(\sE)$, let $\supp(g)$ denote its support. Let $\M(\sE)$ and $\P(\sE)$ denote the set of all finite signed measures and probability measures on $\sE$, respectively. A sequence $\{\mu_n\}$ of finite signed measures on $\sE$ is said to converge with respect to total variation distance (see \cite{HeLa03}) to a finite signed measure $\mu$ if $ \lim_{n\rightarrow\infty} 2\sup_{D \in \cal{E}} |\mu_n(D) - \mu(D)|=0$. A sequence $\{\mu_n\}$ of finite signed measures on $\sE$ is said to converge weakly (see \cite{HeLa03}) to a finite signed measure $\mu$ if $\int_{\sE} g d\mu_n \rightarrow \int_{\sE} g d\mu$ for all bounded and continuous real function $g$ on $\sE$. Let $\sE_1$ and $\sE_2$ be two metric spaces. For any $\mu \in \M(\sE_1 \times \sE_2)$, we denote by $\Proj_{\sE_1}(\mu)(\,\cdot\,) \coloneqq \mu(\,\cdot\,\times \sE_2)$ the marginal of $\mu$ on $\sE_1$. Let $\sE = \prod_{i=1}^N \sE_i$ be a finite product space. For each $j,k = 1,\ldots,N$ with $k < j$, we denote $\sE^{^{[k:j]}} = \prod_{i=k}^j \sE_i$ and $\sE_{-j} = \prod_{i \neq j} \sE_i$. A similar convention also applies to elements of these sets which will be denoted by bold lower case letters. For any set $D$, let $D^c$ denote its complement. Unless otherwise specified, the term `measurable' will refer to Borel measurability in the rest of the paper.

\section{Intrinsic Model for Sequential teams}\label{sub1sec1}



Witsenhausen's {\it intrinsic model} \cite{wit75} for sequential team problems has the following components:
\begin{align}
\bigl\{ (\sX, {\cal X}), \rP, (\sU_i,{\cal U}_i), (\sY_i,{\cal Y}_i), i=1,\ldots,N,\bigr\} \nonumber
\end{align}
where Borel spaces (i.e., Borel subsets of complete and separable metric spaces) $\sX$, $\sU_i$, and $\sY_i$ ($i=1,\ldots,N$) endowed with Borel $\sigma$-algebras denote the state space, and action and observation spaces of Agent~$i$, respectively. Here $N$ is the number of actions taken, and each of these actions is supposed to be taken by an individual agent (hence, an agent with perfect recall can also be regarded as a separate decision maker every time it acts). For each $i$, the observations and actions of Agent~$i$ is denoted by $u_i$ and $y_i$, respectively. The $\sY_i$-valued observation variable for Agent~$i$ is given by $y_i \sim W_i(\,\cdot\,|x,{\bf u}^{^{[1:i-1]}})$, where $W_i$ is a stochastic kernel on $\sY_i$ given $\sX\times \sU^{^{[1:i-1]}}$ \cite[Definition C.1]{HeLa96}. A {\it probability measure} $\rP$ on $(\sX, {\cal X})$ describes the uncertainty on the state variable $x$.

A control strategy $\underline{\gamma}= (\gamma_1, \gamma_2, \dots, \gamma_N)$, also called {\it policy}, is an $N$-tuple of measurable functions such that $u_i = \gamma_i(y_i)$, where $\gamma_i$ is a measurable function from $\sY_i$ to $\sU_i$. Let $\Gamma_i$ denote the set of all admissible policies for Agent~$i$; that is, the set of all measurable functions from $\sY_i$ to $\sU_i$ and let ${\bf \Gamma} = \prod_{k} \Gamma_k$.

We note that the intrinsic model of Witsenhausen uses a set-theoretic characterization; however, for Borel spaces, the model above is equivalent to the intrinsic model for sequential team problems.

Under this intrinsic model, a sequential team problem is {\it dynamic} if the information $y_i$ available to at least one agent~$i$ is affected by the action of at least one other agent~$k\neq i$. A decentralized problem is {\it static}, if the information available at every decision maker is only affected by state of the nature; that is, no other decision maker can affect the information at any given decision maker.


For any $\underline{\gamma} = (\gamma_1, \cdots, \gamma_N)$, we let the
(expected) cost of the team problem be defined by
\[J(\underline{\gamma}) \coloneqq E[c(x,{\bf y},{\bf u})],\]
for some measurable cost function $c: \sX \times \prod_i \sY_i \times \prod_i \sU_i \to [0,\infty)$, where ${\bf u} \coloneqq (u_1,\ldots,u_N) = \underline{\gamma}({\bf y})$ and ${\bf y} \coloneqq (y_1,\ldots,y_N)$.

\begin{definition}\label{Def:TB1}
For a given stochastic team problem, a policy (strategy)
${\underline \gamma}^*:=({\gamma_1}^*,\ldots, {\gamma_N}^*)\in {\bf \Gamma}$ is
an {\it optimal team decision rule} if
\begin{equation}
J({\underline \gamma}^*)=\inf_{{{\underline \gamma}}\in {{\bf \Gamma}}}
J({{\underline \gamma}})=:J^*. \nonumber
\end{equation}
The cost level $J^*$ achieved by this strategy is the {\it optimal team cost}.
\end{definition}

In what follows, the terms \emph{policy}, \emph{measurement}, and \emph{agent} are used synonymously with \emph{strategy}, \emph{observation}, and \emph{decision maker}, respectively.

\subsection{Auxiliary Results}

To make the paper as self-contained as possible, in this section  we review some results in probability theory and functional analysis that will be used in the paper.

The first result is Prokhorov's theorem which gives a sufficient condition for relative compactness in weak topology.

\begin{theorem}{(\cite[Theorem E.6]{HeLa96})}\label{aux1}
A set $\M$ of probability measures on a Borel space $\sE$ is relatively compact with respect to the weak topology if it is tight; that is, for any $\varepsilon>0$ there exists a compact subset $K$ of $\sE$ such that for all $\mu \in \M$, we have $\mu(K) > 1-\varepsilon$, or equivalently, $\mu(K^c) < \varepsilon$.
\end{theorem}

\begin{proposition}{(\cite[Theorem 3.2]{Par67})}\label{aux2}
Let $\mu$ be a probability measure on a Borel space $\sE$. Then $\mu$ is tight.
\end{proposition}

\begin{proposition}{(\cite[Lemma 4.4]{Vil09})}\label{aux3}
Let $\sE_1$ and $\sE_2$ be two Borel spaces. Let $F_1 \subset \P(\sE_1)$ and $F_2 \subset \P(\sE_2)$ be tight subsets of $\P(\sE_1)$ and $\P(\sE_2)$, respectively.
Then the set
\begin{align}
&F \coloneqq  \nonumber \\
&\phantom{xx}\biggl\{ \nu \in \P(\sE_1 \times \sE_2): \mathrm{Proj}_{\sE_1}(\nu) \in F_1 \text{ and } \mathrm{Proj}_{\sE_2}(\nu) \in F_2 \biggr\} \nonumber
\end{align}
is also tight.
\end{proposition}

Before next theorem, we should give the following definition.

\begin{definition}{(\cite[Definition 4.4]{GuYuBaLa15})}\label{aux4}
Let $\sE_1$, $\sE_2$, and $\sE_3$ be Borel spaces. A non-negative measurable function $\varphi: \sE_1 \times \sE_2 \times \sE_3 \rightarrow [0,\infty)$ is in class $\mathrm{IC}(\sE_1,\sE_2)$ if for every $M > 0$ and for every compact set $K \subset \sE_1$, there exists a compact set $L \subset \sE_2$ such that
\begin{align}
\inf_{K \times L^c \times \sE_3} \varphi(e_1,e_2,e_3) \geq M. \nonumber
\end{align}
\end{definition}

Using this definition, we now state the following result.

\begin{theorem}{(\cite[Lemma 4.5]{GuYuBaLa15})}\label{aux5}
Suppose $\varphi: \sE_1 \times \sE_2 \times \sE_3 \rightarrow [0,\infty)$ is in class $\mathrm{IC}(\sE_1,\sE_2)$. Let $m>0$ and $F_1 \subset \P(\sE_1)$ be a tight set of measures. Define
\begin{align}
&F \coloneqq  \nonumber \\
&\phantom{x}\biggl\{ \nu \in \P(\sE_1 \times \sE_2 \times \sE_3): \mathrm{Proj}_{\sE_1}(\nu) \in F_1 \text{ and } \int \varphi d\nu \leq m \biggr\}. \nonumber
\end{align}
Then $\mathrm{Proj}_{\sE_1 \times \sE_2}(F)$ is a tight set of measures.
\end{theorem}

The last result is about convergence of bilinear forms constituting duality between a Banach space and its topological dual, when both terms in bilinear form converges in some sense.

\begin{proposition}\label{aux6}
Let $(\sE,\|\cdot\|)$ be a Banach space with its topological dual $\sE^*$, where the bilinear form that constitutes duality is denoted by $\langle e^*,e \rangle$, $e \in \sE$ and $e^* \in \sE^*$. Suppose $\lim_{n\rightarrow\infty} \| e_n - e \| = 0$ and $e_n^* \rightarrow e^*$ with respect to $w^*$-topology; that is, $\lim_{n\rightarrow\infty} |\langle e^*_n,e \rangle - \langle e^*,e \rangle| = 0$ for all $e \in \sE$. Then we have $\langle e^*_n,e_n \rangle \rightarrow \langle e^*,e \rangle$ as $n\rightarrow\infty$.
\end{proposition}

\begin{proof}
Suppose $\lim_{n\rightarrow\infty} \| e_n - e \| = 0$ and $e_n^* \rightarrow e^*$ with respect to $w^*$-topology. Then we have
\begin{align}
\bigl| \langle e^*_n,e_n \rangle - \langle e^*,e \rangle \bigr| &\leq \bigl| \langle e^*_n,e_n \rangle - \langle e^*_n,e \rangle \bigr| + \bigl| \langle e^*_n,e \rangle - \langle e^*,e \rangle \bigr| \nonumber \\
&\leq \|e_n^*\| \|e_n-e\| + \bigl| \langle e^*_n,e \rangle - \langle e^*,e \rangle \bigr|. \nonumber
\end{align}
The second term in the last expression converges to zero as $e_n^* \rightarrow e^*$ with respect to $w^*$-topology. Note that $\sup_{n}  \|e_n^*\| < \infty$ by Uniform Boundedness Principle \cite[Theorem 5.13]{Fol99}. Hence the first term in the last expression also converges to zero as $\lim_{n\rightarrow\infty} \| e_n - e \| = 0$.
\end{proof}

\section{Existence of the Optimal Strategy for Static Team Problems}\label{sec2}

In this section, we show the existence of optimal strategy for static teams. Recall that $\bigl(\sX,{\cal X},\rP)$ is a probability space representing the state space, where $\sX$ is a Borel space and ${\cal X}$ is its Borel $\sigma$-algebra. We consider an $N$-agent static team problem in which Agent $i$ ($i=1,\ldots,N$) observes a random variable $y_i$ and takes an action $u_i$, where $y_i$ takes values in a Borel space $\sY_i$ and $u_i$ takes values in a Borel space $\sU_i$. Given any state realization $x$, the random variable $y_i$ has a distribution $W_i(\,\cdot\,|x)$; that is, $W_i(\,\cdot\,|x)$ is a stochastic kernel on $\sY_i$ given $\sX$.

The team cost function $c$ is a non-negative function of the state, observations, and actions; that is, $c: \sX \times \sY \times \sU \rightarrow [0,\infty)$, where $\sY \coloneqq \prod_{i=1}^N \sY_i$ and $\sU \coloneqq \prod_{i=1}^N \sU_i$. To prove the existence of team-optimal policies, we enlarge the space of strategies where each agent can also apply \emph{randomized} strategies; that is, for Agent $i$, the set of strategies $\Gamma_i$ is defined as
\begin{align}
\Gamma_i \coloneqq \biggl\{\gamma_i: \text{$\gamma_i(\,\cdot\,|y_i)$ is a stochastic kernel on $\sU_i$ given $\sY_i$} \biggr\}. \nonumber
\end{align}
We first prove the existence of optimal \emph{randomized} strategy. Then, using Blackwell's irrelevant information theorem \cite{Bla64}, we deduce that the optimal strategy can be chosen deterministic which therefore solves the problem for the original setup.

Recall that ${\bf \Gamma} = \prod_{i=1}^N \Gamma_i$. Then, the cost of the team $J: {\bf \Gamma} \rightarrow [0,\infty)$ is given by
\begin{align}
J(\underline{\gamma}) = \int_{\sX \times \sY \times \sU} c(x,{\bf y},{\bf u}) \underline{\gamma}(d\bu|\by) \rP(dx,d{\bf y}), \nonumber
\end{align}
where $\underline{\gamma}(d\bu|\by) = \prod_{i=1}^N \gamma_i(du_i|y_i)$. Here, with an abuse of notation, $\rP(dx,d{\bf y}) \coloneqq ~\rP(dx) \prod_{i=1}^N W_i(dy_i|x)$ denotes the joint distribution of the state and observations. Therefore, we have
\begin{align}
J^* = \inf_{\underline{\gamma} \in {\bf \Gamma} } J(\underline{\gamma}). \nonumber
\end{align}

For any strategy $\underline{\gamma}$, we let $Q_{\underline{\gamma}} = \underline{\gamma}(d\bu|\by) \rP(dx,d{\bf y})$ denote the probability measure induced on $\sX \times \sY \times \sU$. In the literature, $Q_{\underline{\gamma}}$ is called \emph{strategic} measure.

In this section, we impose the following assumptions.

\begin{assumption}\label{as1}
\begin{itemize}
\item [ ]
\item [(a)] The cost function $c$ is lower semi-continuous.
\item [(b)] $\sX$, $\sY_i$, and $\sU_i$ ($i=1,\ldots,N$) are locally compact.
\item [(c)] For all $i$, $W_i: \sX \rightarrow \P(\sY_i)$ is continuous with respect to the total variation distance.
\item [(d)] For all $i$, $W_i(dy_i|x) = q_i(y_i,x) \mu_i(dy_i)$ for some probability measure $\mu_i$ on $\sY_i$.
\end{itemize}
\end{assumption}

\begin{remark}
Note that, for all $i=1,\ldots,N$, if $q_i$ is continuous in $x$ and $q_i(y_i,x) \leq w_i(y_i)$ for some $\mu_i$-integrable $w_i$, then Assumption~\ref{as1}-(c) holds. Indeed, let $x_n \rightarrow x$ in $\sX$. Then we have
\begin{align}
\bigl\|W_i(\,\cdot\,|x_n) - W_i(\,\cdot\,|x)&\bigr\|_{TV} \nonumber \\
&= \int_{\sY_i} \bigl| q_i(y_i,x_n) - q_i(y_i,x) \bigr| \mu_i(dy_i). \nonumber
\end{align}
The last expression goes to $0$ as $n\rightarrow\infty$ by dominated convergence theorem.
\end{remark}

\begin{remark}
One common approach that is used in the literature \cite{GuYuBaLa15,YuSa17} to show the existence of team-optimal policies is strategic measure approach. In this approach, one first identifies a topology on the set of strategic measures $\Xi \coloneqq \{Q_{\underline{\gamma}}: \underline{\gamma} \in {\bf \Gamma}\}$ (in general, weak topology) and then proves the relative compactness of $\Xi$ along with lower semi-continuity of the cost function $J$ with respect to this topology. Then, if $\Xi$ is closed with respect to this topology, then one can deduce the existence of an optimal policy via Weierstrass Extreme Value Theorem. The main problem in this approach is to prove the closeness of $\Xi$, because convergence with respect to the topology defined on $\Xi$ does not in general preserve the statistical independence of the actions given the observations; that is, in the limiting strategic measure, action $u_i$ of Agent~$i$ may depend on observation $y_j$ of Agent~$j$ which is prohibited in the original problem (see, e.g., \cite[Theorem 2.7]{YuSa17}). Hence, to overcome this obstacle, in this paper we directly introduce a topology on the set policies ${\bf \Gamma}$ instead of the set of strategic measures $\Xi$. By this way, in the limiting measure, we can preserve the statistical independence of actions given the observations.
\end{remark}

\subsection{Topology on the Set of Policies ${\bf \Gamma}$}\label{topology}

In this section we introduce a topology on the set of policies ${\bf \Gamma}$, which will be used to obtain the existence of team-optimal policies. To this end, we first identify a topology on $\Gamma_i$ for each $i=1,\ldots,N$. Fix any $i \in \{1,\ldots,N\}$.

Recall that we denote by $C_0(\sU_i)$, $\M(\sU_i)$, and $\P(\sU_i)$ the set of real continuous functions vanishing at infinity on $\sU_i$, the set of finite signed measures on $\sU_i$, and the set of probability measures on $\sU_i$, respectively. For any $g \in C_0(\sU_i)$, let $\|g\|_{\infty} \coloneqq \sup_{u \in \sU_i} |g(u)|$ which turns $(C_0(\sU_i),\|\,\cdot\,\|_{\infty})$ into a Banach space. Let $\|\,\cdot\,\|_{TV}$ denote the total variation norm on $\M(\sU_i)$ which turns $(\M(\sU_i),\|\,\cdot\,\|_{TV})$ into a Banach space.

\begin{theorem}\cite[Theorem 7.17]{Fol99}
For any $\nu \in \M(\sU_i)$ and $g \in C_0(\sU_i)$, let $I_{\nu}(g) \coloneqq \langle g, \nu \rangle$, where
\begin{align}
\langle g, \nu \rangle \coloneqq \int_{\sU_i} g d\nu. \nonumber
\end{align}
Then the map $\nu \mapsto I_{\nu}$ is an isometric isomorphism from $\M(\sU_i)$ to $C_0(\sU_i)^*$. Hence, we can identify $C_0(\sU_i)^*$ with $\M(\sU_i)$.
\end{theorem}

A function $\gamma: \sY_i \rightarrow \M(\sU_i)$ is called $w^*$-measurable \cite[p. 18]{CeMe97} if the mapping $\sY_i \ni y \mapsto \langle g, \gamma(y) \rangle \in \R$ is measurable for all $g \in C_0(\sU_i)$. Let $\rL\bigl(\mu_i,\M(\sU_i)\bigr)$ denote the set of all such functions. Then, we define the following set
\begin{align}
&\rL_{\infty}\bigl(\mu_i,\M(\sU_i)\bigr) \nonumber \\
&\phantom{}\coloneqq \biggl\{ \gamma \in \rL\bigl(\mu_i,\M(\sU_i)\bigr): \|\gamma\|_{\infty} \coloneqq \ess \sup_{y \in \sY_i} \|g(y)\|_{TV} <\infty \biggr\}, \nonumber
\end{align}
where $\ess \sup$ is taken with respect to the measure $\mu_i$. Recall that $\mu_i$ is the reference probability measure in Assumption~\ref{as1}-(d) for the observation channel $W_i$.

A function $f: \sY_i \rightarrow C_0(\sU_i)$ is said to be simple if there exists $g_1,\ldots,g_n \in C_0(\sU_i)$ and $E_1,\ldots,E_n \in \B(\sY_i)$ such that $f = \sum_{i=1}^n g_i 1_{E_i}$. Define the Bochner integral of $f$ with respect to $\mu_i$ as
\begin{align}
\int_{\sY_i} f(y) \mu_i(dy) \coloneqq \sum_{i=1}^n g_i \mu_i(E_i). \nonumber
\end{align}
A function $f:\sY_i \rightarrow C_0(\sU_i)$ is  said to be strongly measurable, if there exists a sequence $\{f_n\}$ of simple functions with $\lim_{n\rightarrow\infty} \|f_n(y)-f(y)\|_{\infty}=0$ $\mu_i$-almost everywhere. The strongly measurable function $f$ is Bochner-integrable \cite{DiUh77} if $\int_{\sY_i} \|f(y)\|_{\infty} \mu_i(dy) < \infty$. In this case, the integral is given by
\begin{align}
\int_{\sY_i} f(y) \mu_i(dy) = \lim_{n\rightarrow\infty} \int_{\sY_i} f_n(y) \mu_i(dy), \nonumber
\end{align}
where $\{f_n\}$ is the sequence of simple functions which approximates $f$. Let $L_1\bigl(\mu_i,C_0(\sU_i)\bigr)$ denote the set of all Bochner-integrable functions from $\sY_i$ to $C_0(\sU_i)$ endowed with the norm
\begin{align}
\|f\|_1 \coloneqq \int_{\sY_i} \|f(y)\|_{\infty} \mu_i(dy_i). \nonumber
\end{align}
Then, we have the following theorem.

\begin{theorem}\cite[Theorem 1.5.5, p. 27]{CeMe97}\label{duality2}
For any $\gamma \in \rL_{\infty}\bigl(\mu_i,\M(\sU_i)\bigr)$ and $f \in L_1\bigl(\mu_i,C_0(\sU_i)\bigr)$, let
\begin{align}
T_{\gamma}(f) \coloneqq \int_{\sY_i} \langle f(y), \gamma(y) \rangle \mu_i(dy). \nonumber
\end{align}
Then the map $\gamma \mapsto T_{\gamma}$ is an isometric isomorphism from $\rL_{\infty}\bigl(\mu_i,\M(\sU_i)\bigr)$ to $L_1\bigl(\mu_i,C_0(\sU_i)\bigr)^*$. Hence, we can identify $L_1\bigl(\mu_i,C_0(\sU_i)\bigr)^*$ with $\rL_{\infty}\bigl(\mu_i,\M(\sU_i)\bigr)$.
\end{theorem}

By Theorem~\ref{duality2}, we equip $\rL_{\infty}\bigl(\mu_i,\M(\sU_i)\bigr)$ with $w^*$-topology induced by $L_1\bigl(\mu_i,C_0(\sU_i)\bigr)$; that is, it is the smallest topology on $\rL_{\infty}\bigl(\mu_i,\M(\sU_i)\bigr)$ for which the mapping
\begin{align}
\rL_{\infty}\bigl(\mu_i,\M(\sU_i)\bigr) \ni \gamma \mapsto T_{\gamma}(f) \in \R \nonumber
\end{align}
is continuous for all $f \in L_1\bigl(\mu_i,C_0(\sU_i)\bigr)$ \cite{Fol99}. We write $\gamma_n \rightharpoonup^* \gamma$, if $\gamma_n$ converges to $\gamma$ in $\rL_{\infty}\bigl(\mu_i,\M(\sU_i)\bigr)$ with respect to $w^*$-topology. Note that, for this topology, we have been in part inspired by the topology introduced in \cite[Section 2.4]{BoArGh12}, where in this work, a similar topology is introduced for randomized Markov policies to study continuous-time stochastic control problems with average cost optimality criterion (see \cite{Bor89} for another construction of a topology on Markov policies).

\begin{lemma}\label{kernel}
Suppose $\gamma \in \rL_{\infty}\bigl(\mu_i,\M(\sU_i)\bigr)$ such that $\gamma(y) \in \P(\sU_i)$ $\mu_i$-a.e.. Then, for all $D \in {\cal U}_i$, the mapping $y \mapsto \gamma(y)(D)$ is measurable. Hence, $\gamma$ is a stochastic kernel.
\end{lemma}

\begin{proof}
Note first that the mapping $\sY_i \ni y \mapsto \langle g, \gamma(y) \rangle \in \R$ is measurable for all real, continuous, and bounded $g$ on $\sU_i$, because any such $g$ can be approximated pointwise by $\{g_n\}_{n\geq1} \subset C_0(\sU_i)$ satisfying $ \|g_n\|_{\infty} \leq \|g\|_{\infty}$ for all $n$. Moreover, for any closed set $F \subset \sU_i$, one can approximate pointwise the indicator function $1_F$ by continuous and bounded functions $h_n \coloneqq \max\bigl(1 - n d_{\sU_i}(x,C) \bigr)$, where $d_{\sU_i}$ is the metric on $\sU_i$ and $d_{\sU_i}(x,C) \coloneqq \inf_{y \in C} d_{\sU_i}(x,y)$. This implies that the mapping $y \mapsto \gamma(y)(F)$ is measurable for all closed set $F$ in $\sU_i$. Then the result follows by \cite[Proposition 7.25]{BeSh78}.
\end{proof}

By Lemma~\ref{kernel}, we have
\begin{align}
\Gamma_i = \biggl\{ \gamma \in \rL_{\infty}\bigl(\mu_i,\M(\sU_i)\bigr): \gamma(y) \in \P(\sU_i) \text{ } \mu_i-\text{a.e.} \biggr\}. \nonumber
\end{align}

Since $\P(\sU_i)$ is bounded in $\M(\sU_i)$, by Banach-Alaoglu Theorem \cite[Theorem 5.18]{Fol99}, $\Gamma_i$ is  relatively compact with respect to $w^*$-topology. Since $C_0(\sU_i)$ is separable, then by \cite[Lemma 1.3.2]{HeLa03}, $\Gamma_i$ is also relatively sequentially compact.

Note that $\Gamma_i$ is not closed with respect to $w^*$-topology. Indeed, let $\sY_i = \sU_i = \R$. Define $\gamma_n(y_i)(\,\cdot\,) \coloneqq \delta_n(\,\cdot\,)$ and $\gamma(y_i)(\,\cdot\,) \coloneqq 0(\,\cdot\,)$, where $0(\,\cdot\,)$ denotes the degenerate measure on $\sU_i$; that is, $0(D)=0$ for all $D \in \B(\R)$. Let $g \in L_1\bigl(\mu_i,C_0(\sU_i)\bigr)$. Then we have
\begin{align}
\lim_{n\rightarrow\infty}& \int_{\sY_i} \langle g(y), \gamma_n(y) \rangle \mu_i(dy) = \lim_{n\rightarrow\infty} \int_{\sY_i} g(y)(n) \mu_i(dy) \nonumber \\
&= \int_{\sY_i} \lim_{n\rightarrow\infty} g(y)(n) \mu_i(dy) \text{ (as $\|g(y)\|_{\infty}$ is $\mu_i$-integrable)} \nonumber \\
&= 0 \text{ (as $g(y) \in C_0(\sU_i)$)}. \nonumber
\end{align}
Hence, $\gamma_n \rightharpoonup^* \gamma$. But, $\gamma \notin \Gamma_i$, and so, $\Gamma_i$ is not closed.

In the remainder of this section, $\Gamma_i$ is equipped with this topology. In addition, ${\bf \Gamma}$ has the product topology induced by these $w^*$-topologies; that is, $\underline{\gamma}_n$ converges to $\underline{\gamma}$ in ${\bf \Gamma}$ with respect to the product topology if and only if $\gamma_{n,i} \rightharpoonup^* \gamma_i$ for all $i=1,\ldots,N$. In this case we write $\underline{\gamma}_n \rightharpoonup^* \underline{\gamma}$. Note that ${\bf \Gamma}$ is sequentially relatively compact under this topology.

\subsection{Existence of Team-Optimal Policies}

In this section, using the topology introduced in Section~\ref{topology}, we prove the existence of an optimal policy under the Assumption~\ref{as1} and the assumption below. For any $L>0$, we define
\begin{align}
{\bf \Gamma}_L &\coloneqq \biggl\{\underline{\gamma} \in {\bf \Gamma}: J(\underline{\gamma}) < J^* + L  \biggr\} \nonumber \\
\intertext{and}
S_{L} &\coloneqq \biggl\{\lambda \in \P(\sX \times \sY \times \sU): \nonumber \\
&\phantom{xx}\lambda(dx,d\by,d\bu) = \rP(dx) \prod_{i=1}^N \gamma_i(du_i|y_i) \mu_i(dy_i), \underline{\gamma} \in {\bf \Gamma}_L \biggr\}. \nonumber
\end{align}
For each $i=1,\ldots,N$, we define $S_L^i \coloneqq \Proj_{\sY_i \times \sU_i}(S_L)$.

\begin{assumption}\label{as2}
For some $L>0$, $S_L^i$ is tight for $i=1,\ldots,N$.
\end{assumption}

Before we continue with the proof, we will give several conditions that imply Assumption~\ref{as2}.

\begin{theorem}
Suppose \emph{either }of the following conditions hold:
\begin{itemize}
\item [(i)] $\sU_i$ is compact for all $i$.
\item [(ii)] For non-compact case, we assume
\begin{itemize}
\item [(a)] The cost function $c$ is in class $\mathrm{IC}(\sX \times \sY \times \sU^{^{[1:j-1]}},\sU_j)$, for all $j$.
\item [(b)] For all $j$, $q_j > 0$ and $q_j$ is lower semi-continuous.
\end{itemize}
\end{itemize}
Then, Assumption~\ref{as2} holds.
\end{theorem}

\begin{proof}
(i): Note that the marginal on $\sY_i$ of any measure in $S_L^i$ is $\mu_i$. Since $\mu_i$ is tight by Proposition~\ref{aux2} and $\P(\sU_i)$ is tight by the compactness of $\sU_i$,  $S_L^i$ is also tight by Proposition~\ref{aux3}.\\
\noindent(ii):
We define $\tilde{c}(x,\by,\bu) \coloneqq c(x,\by,\bu) \prod_{i=1}^N q_i(y_i,x)$. Since, for all $j$, $q_j$ is lower semi-continuous and strictly greater than $0$, for any compact set $K \subset \sX \times \sY \times \sU^{^{[1:j-1]}}$, we have $\inf_{K} \prod_i^N q_i(y_i,x) > 0$. This implies that $\tilde{c}$ is also in class $\mathrm{IC}(\sX \times \sY \times \sU^{^{[1:j-1]}},\sU_j)$ for $j=1,\ldots,N$. Then, by Theorem~\ref{aux5}, one can inductively prove that $\Proj_{\sX \times \sY \times \sU^{^{[1:j]}}}(S_L)$ is tight. Indeed, let $j=1$. Then $\tilde{c}$ is in $\mathrm{IC}(\sX \times \sY,\sU_1)$ and
\begin{align}
&S_L \subset \biggl\{ \lambda \in \P(\sX \times \sY \times \sU): \mathrm{Proj}_{\sX \times \sY}(\lambda)(dx,d\by) \nonumber \\
&\phantom{xxxxxxxx}= \rP(dx) \prod_i^N \mu_i(dy_i) \text{ and } \int \tilde{c} d\lambda \leq J^*+L \biggr\}. \nonumber
\end{align}
But since $\rP(dx) \prod_i^N \mu_i(dy_i)$ is tight, by Theorem~\ref{aux5}, $\Proj_{\sX \times \sY \times \sU_1}(S_L)$ is also tight. Suppose the assertion is true for $j$ and consider $j+1$. Note that $\tilde{c}$ is in $\mathrm{IC}(\sX \times \sY \times \sU^{^{[1:j]}}, \sU_{j+1})$ and
\begin{align}
&S_L \subset \biggl\{ \lambda \in \P(\sX \times \sY \times \sU): \mathrm{Proj}_{\sX \times \sY \times \sU^{^{[1:j]}}}(\lambda) \in \nonumber \\
&\phantom{xxxxxxxx} \mathrm{Proj}_{\sX \times \sY \times \sU^{^{[1:j]}}}(S_L) \text{ and } \int \tilde{c} d\lambda \leq J^*+L \biggr\}. \nonumber
\end{align}
Since $\mathrm{Proj}_{\sX \times \sY \times \sU^{^{[1:j]}}}(S_L)$ is tight by the induction hypothesis, $\mathrm{Proj}_{\sX \times \sY \times \sU^{^{[1:j+1]}}}(S_L)$ is also tight
by Theorem~\ref{aux5}. This completes the proof of assertion. But this result implies that $S_L^j$ is also tight for all $j$.
\end{proof}

Recall that $C_c(\sX \times \sY \times \sU)$ denotes the set of real continuous functions on $\sX \times \sY \times \sU$ with compact support. For any $g \in C_c(\sX \times \sY \times \sU)$, we define
\begin{align}
J_g(\underline{\gamma}) = \int_{\sX \times \sY \times \sU} g(x,{\bf y},{\bf u}) \underline{\gamma}(d\bu|\by) \rP(dx,d{\bf y}). \nonumber
\end{align}
We first prove the following result.

\begin{theorem}\label{main1}
Suppose that $\underline{\gamma}^{(n)} \rightharpoonup^* \underline{\gamma}$ as $n\rightarrow\infty$ and $g \in C_c(\sX \times \sY \times \sU)$. Then we have
\begin{align}
\lim_{n\rightarrow\infty} \bigl| J_g(\underline{\gamma}^{(n)}) - J_g(\underline{\gamma}) \bigr| = 0. \nonumber
\end{align}
\end{theorem}

\begin{proof}
Fix any $g \in C_c(\sX \times \sY \times \sU)$. Then by Stone-Weierstrass Theorem \cite[Lemma 6.1]{Lan93}, $g$ can be uniformly approximated by functions of the form
\begin{align}
\sum_{j=1}^k r_j \prod_{i=1}^N f_{j,i} g_{j,i}, \nonumber
\end{align}
where $r_j \in C_c(\sX)$, $f_{j,i} \in C_c(\sY_i)$, and $g_{j,i} \in C_c(\sU_i)$ for each $j=1,\ldots,k$ and $i=1,\ldots,N$. This implies that it is sufficient to prove the result for functions of the form $r \prod_{i=1}^N f_{i} g_{i}$, where $r \in C_c(\sX)$, $f_{i} \in C_c(\sY_i)$, and $g_{i} \in C_c(\sU_i)$ for $i=1,\ldots,N$.  Therefore, in the sequel, we assume that $g = r \prod_{i=1}^N f_{i} g_{i}$.

Let $K \coloneqq \supp(r)$ which is a compact subset of $\sX$ as $r \in C_c(\sX)$. Then we have
\begin{align}
&\bigl| J_g(\underline{\gamma}^{(n)}) - J_g(\underline{\gamma}) \bigr|  \nonumber \\
&\leq  \bigl| J_g(\gamma_1^{(n)},\ldots,\gamma_N^{(n)}) - J_g(\gamma_1^{(n)},\ldots,\gamma_{N-1}^{(n)},\gamma_N) \bigr| \nonumber \\
&+  \bigl| J_g(\gamma_1^{(n)},\ldots,\gamma_{N-1}^{(n)},\gamma_N) - J_g(\gamma_1^{(n)},\ldots,\gamma_{N-2}^{(n)},\gamma_{N-1},\gamma_N) \bigr| \nonumber \\
&\phantom{x}\vdots \nonumber \\
&+  \bigl| J_g(\gamma_1^{(n)},\gamma_2,\ldots,\gamma_N) - J_g(\gamma_1,\ldots,\gamma_N) \bigr| \nonumber \\
&\eqqcolon \sum_{j=1}^N l_j^{(n)}. \nonumber
\end{align}
Let us consider the $j^{th}$ term in the above expression. Define the probability measure $T_{-j}$ on $\sX \times \sY_{-j} \times \sU_{-j}$ and real function $g_{-j}$ on $\sX \times \sY_{-j} \times \sU_{-j}$ as follows:
\begin{align}
&T_{-j} \nonumber \coloneqq \biggl( \prod_{i=j+1}^N \gamma_i(du_i|y_i) q_i(y_i,x) \mu_i(dy_i) \biggr) \times \nonumber \\
&\phantom{xxxxxxxxxxxxxx}\biggl( \prod_{i=1}^{j-1} \gamma_i^{(n)}(du_i|y_i) q_i(y_i,x) \mu_i(dy_i) \biggr) \rP(dx) \nonumber
&\intertext{and}
&g_{-j} \coloneqq r \prod_{i \neq j} f_i g_i. \nonumber
\end{align}
Then the $j^{th}$ term can be written as
\begin{align}
l_j^{(n)} &= \biggl| \int g_{-j} \biggl( \int f_j g_j q_j d\gamma_j^{(n)}\otimes\mu_j \biggr) dT_{-j} \nonumber \\
&\phantom{xxxxxxxxxxx} - \int g_{-j} \biggl( \int f_j g_j q_j d\gamma_j\otimes\mu_j \biggr) dT_{-j} \biggr|. \nonumber
\end{align}

Define, for each $x \in \sX$, the function
\begin{align}
b_x(y_j,u_j) \coloneqq f_j(y_j) g_j(u_j) q_j(y_j,x). \nonumber
\end{align}
One can prove that any $b_x$ is in $L_1(\mu_j,C_0(\sU_j))$; that is $b_x(y_j,\,\cdot\,) \in C_0(\sU_j)$ almost all $y_j \in \sY_j$ and $b_x$ can be approximated by simple functions. We will prove that the set $\{b_x\}_{x \in K} \subset L_1(\mu_j,C_0(\sU_j))$ is totally bounded. Indeed, let $x,\tilde{x} \in K$. Then
\begin{align}
&\|b_x - b_{\tilde{x}}\|_1 \coloneqq \int_{\sY_j} \sup_{u_j \in \sU_j} \bigl| f_j(y_j) g_j(u_j) q_j(y_j,x) \nonumber \\
&\phantom{xxxxxxxxxxxxxxxxxx}- f_j(y_j) g_j(u_j) q_j(y_j,\tilde{x}) | \mu_j(dy_j) \nonumber \\
&\leq \|f_j\|_{\infty} \|g_j\|_{\infty} \int_{\sY_j} \bigl| q_j(y_j,x) - q_j(y_j,\tilde{x}) | \mu_j(dy_j) \nonumber \\
&= \|f_j\|_{\infty} \|g_j\|_{\infty} \bigl\|W_j(\,\cdot\,|x) - W_j(\,\cdot\,|\tilde{x})\bigr\|_{TV}. \label{unif}
\end{align}
Since $W_j$ is assumed to be continuous with respect to the total variation distance, the set $\{b_x\}_{x \in K}$ is totally bounded; that is, for any $\varepsilon > 0$, there exists a finite number of points $x_1,\ldots,x_n \in K$ such that
\begin{align}
\{  b_x  \}_{x \in K} \subset \bigcup_{i=1}^n B_1(b_{x_i},\epsilon), \nonumber
\end{align}
where $B_1(b_{x},\varepsilon) \coloneqq \{b \in L_1(\mu_j,C_0(\sU_j)): \|b-b_x\|_1 \leq \varepsilon\}$. Indeed, fix any $\varepsilon > 0$. Note first that the observation kernel $W_j: K \rightarrow \P(\sY_j)$ is uniformly continuous since $K$ is compact. Hence for any $\epsilon > 0$, one can find $\delta > 0$ such that if $d_{\sX}(x,y) < \delta$, then
$\|W_j(\,\cdot\,|x) - W_j(\,\cdot\,|y)\|_{TV} < \epsilon$. For this $\delta > 0 $, since $K$ is compact, one can find a finite number of points $x_1,\ldots,x_n \in K$ such that
\begin{align}
K \subset \bigcup_{i=1}^n B(x_i,\delta),  \nonumber
\end{align}
where $B(x,\delta) \coloneqq \{y \in \sX: d_{\sX}(x,y) \leq \delta\}$. But this implies that
\begin{align}
\{  b_x  \}_{x \in K} \subset \bigcup_{i=1}^n B_1\bigl(b_{x_i},\epsilon \|f_j\|_{\infty} \|g_j\|_{\infty} \bigr) \nonumber
\end{align}
because if $b_x$ is some element in the set  $\{  b_x  \}_{x \in K}$, then $x$ is in $B(x_i,\delta)$ for some $i$; that is, $d_{\sX}(x,x_i) < \delta$. This implies from uniform continuity that $\|W_j(\,\cdot\,|x) - W_j(\,\cdot\,|x_i)\|_{TV} < \epsilon$, and so,  by (\ref{unif}), we have $\|b_{x_i} - b_x\|_1 < \epsilon \|f_j\|_{\infty} \|g_j\|_{\infty}$.
By choosing $\epsilon = \varepsilon/(\|f_j\|_{\infty} \|g_j\|_{\infty})$, we complete the proof of the assertion.

Using the total boundedness of the set $\{b_x\}_{x \in K}$, we prove the following:
\begin{align}
\lim_{n\rightarrow\infty} \sup_{x \in K} \bigl| \langle b_x, \gamma_j^{(n)} \rangle - \langle b_x, \gamma_j \rangle \bigr| = 0. \label{des1}
\end{align}
Suppose (\ref{des1}) is not true. Then there exists a sub-sequence $\{\gamma_j^{n_k}\}$ of $\{\gamma_j^{n}\}$ such that, for all $k$, we have
\begin{align}
\sup_{x \in K} \bigl| \langle b_x, \gamma_j^{(n_k)} \rangle - \langle b_x, \gamma_j \rangle \bigr| > 0
\end{align}
Suppose $\{\varepsilon_k\}$ be a sequence of positive real numbers converging to zero. For each $k$, let $x_k \in K$ be such that
\begin{align}
\bigl| \langle b_{x_k}, \gamma_j^{(n_k)} \rangle - \langle b_{x_k}, \gamma_j \rangle \bigr| &> \sup_{x \in K} \bigl| \langle b_x, \gamma_j^{(n_k)} \rangle - \langle b_x, \gamma_j \rangle \bigr| - \varepsilon_k \nonumber \\
&> 0. \label{cont}
\end{align}
Since $\{b_{x_{k}}\}$ is totally bounded, there exists a subsequence $\{b_{x_{k_l}}\}$ such that
\begin{align}
b_{x_{k_l}} \rightarrow b \in L_1(\mu_j,C_0(\sU_j)) \text{ in $L_1$-norm.} \nonumber
\end{align}
Then, by Proposition~\ref{aux6}, we have
\begin{align}
\langle b_{x_{k_l}}, \gamma_j^{(n_{k_l})} \rangle \rightarrow \langle b, \gamma_j \rangle \label{important}
\end{align}
as $\gamma_j^{(n_{k_l})} \rightharpoonup^* \gamma_j$. In addition, we also have $\langle b_{x_{k_l}}, \gamma_j \rangle \rightarrow \langle b, \gamma_j \rangle$. Hence,
\begin{align}
\lim_{l\rightarrow\infty} \bigl| \langle b_{x_{k_l}}, \gamma_j^{(n_{k_l})} \rangle - \langle b_{x_{k_l}}, \gamma_j \rangle \bigr| = 0. \nonumber
\end{align}
This contradicts with (\ref{cont}), and so, (\ref{des1}) is true.


Note that we have
\begin{align}
\int f_j g_j q_j d\gamma_j^{(n)}\otimes\mu_j - \int f_j g_j q_j d&\gamma_j\otimes\mu_j \nonumber \\
&= \langle b_x, \gamma_j^{(n)} \rangle - \langle b_x, \gamma_j \rangle. \nonumber
\end{align}
Therefore, we can bound $l_j^{(n)}$ as
\begin{align}
l_j^{(n)} &\leq \int g_{-j} \bigl| \langle b_x, \gamma_j^{(n)} \rangle - \langle b_x, \gamma_j \rangle \bigr| dT_{-j} \nonumber \\
&\leq \|g_{-j}\|_{\infty} \sup_{x\in K} \bigl| \langle b_x, \gamma_j^{(n)} \rangle - \langle b_x, \gamma_j \rangle \bigr|. \nonumber
\end{align}
Note that the last term converges to zero as $n\rightarrow\infty$ by (\ref{des1}). Since $j$ is arbitrary, $l_j^{(n)} \rightarrow 0$ as $n\rightarrow\infty$ for all $j=1,\ldots,N$. This implies that $J_g(\underline{\gamma}^{(n)}) \rightarrow J_g(\underline{\gamma})$ which completes the proof.
\end{proof}

The following theorem is the main result of this section which establishes the existence of team-optimal policies.

\begin{theorem}\label{main2}
Suppose Assumptions~\ref{as1} and \ref{as2} hold. Then, there exists $\underline{\gamma}^{*} \in {\bf \Gamma}$ which is optimal. Moreover, this strategy can be chosen deterministic; that is, $u_i=\gamma_i(y_i)$ for some measurable $\gamma_i:\sY_i\rightarrow \sU_i$, via Blackwell's irrelevant information theorem \cite[p. 457]{YuBa13}.
\end{theorem}

\begin{proof}
Suppose $\{\underline{\gamma}_n\}$ be a minimizing sequence in ${\bf \Gamma}_L$; that is, for each $n$, we have $J(\underline{\gamma}_n) < J^* + \epsilon(n)$, where $\epsilon(n)\rightarrow0$ as $n\rightarrow\infty$. Since ${\bf \Gamma}$ is relatively sequentially compact with respect to $w^*$-topology, there exists a subsequence $\{\underline{\gamma}_{n_k}\}$ of $\{\underline{\gamma}_n\}$ such that
\begin{align}
\underline{\gamma}_{n_k} \rightharpoonup^* \underline{\gamma}^*, \nonumber
\end{align}
for some $\underline{\gamma} \in \prod_{i=1}^N \rL_{\infty}\bigl(\mu_i,\M(\sU_i)\bigr)$ (recall that ${\bf \Gamma}$ is not closed with respect to $w^*$-topology). As $\sX \times \sY \times \sU$ is locally compact, one can find a sequence of $\{c_m\} \subset C_c(\sX \times \sY \times \sU)$ such that $0\leq c_1 \leq c_2 \leq \ldots \leq c_m \leq \ldots \leq c$ and $c_m \rightarrow c$ pointwise (see the proof of \cite[Proposition 1.4.18]{HeLa03}). Then we have
\begin{align}
J^* = \liminf_{k\rightarrow\infty} J(\underline{\gamma}^{n_k}) &=
\liminf_{k\rightarrow\infty} \int c dQ_{\underline{\gamma}^{n_k}} \nonumber \\
&= \liminf_{k\rightarrow\infty} \sup_{m\geq1} \int c_m dQ_{\underline{\gamma}^{n_k}} \nonumber \\
&\geq  \sup_{m\geq1} \liminf_{k\rightarrow\infty} \int c_m dQ_{\underline{\gamma}^{n_k}}  \nonumber \\
&=\sup_{m\geq1} \int c_m dQ_{\underline{\gamma}^*} \text{ }(\text{by Theorem~\ref{main1}}) \nonumber \\
&=\int c dQ_{\underline{\gamma}^*}. \nonumber
\end{align}
This implies that if $\gamma_i^* \in \P(\sU_i)$ $\mu_i$-a.e. for all $i$, then $\underline{\gamma}^*$ is the optimal policy. We now prove $\gamma_i^* \in \P(\sU_i)$ $\mu_i$-a.e. for all $i$.

Fix any $i$. Note that the sequence $\{\gamma_i^{n_k}\otimes\mu_i\}$ is tight as it is a subset of $S_L^i$. Thus, by Theorem~\ref{aux1}, there exists a further subsequence, denoted for simplicity by $\{\gamma_i^{l}\otimes\mu_i\}$, that converges weakly to some $\lambda \in \P(\sY_i \times \sU_i)$. Suppose $g \in C_c(\sY_i \times \sU_i)$, and so, $g(y_i,\,\cdot\,) \in L_1\bigl(\mu_i,C_0(\sU_i)\bigr)$. Since $\gamma_i^l \rightharpoonup^* \gamma^*$, we have
\begin{align}
&\lim_{l\rightarrow\infty} \int_{\sY_i \times \sU_i} g(u_i,y_i) \gamma_i^l(du_i|y_i) \mu_i(dy_i) \nonumber \\
&\phantom{xxxxxxxxxxxx}= \int_{\sY_i \times \sU_i} g(u_i,y_i) \gamma_i^*(du_i|y_i) \mu_i(dy_i). \nonumber
\end{align}
Similarly, by weak convergence of $\gamma_i^{l}\otimes\mu_i$ to $\lambda$, we have
\begin{align}
&\lim_{l\rightarrow\infty} \int_{\sY_i \times \sU_i} g(u_i,y_i) \gamma_i^l(du_i|y_i) \mu_i(dy_i) \nonumber \\
&\phantom{xxxxxxxxxxxx}= \int_{\sY_i \times \sU_i} g(u_i,y_i) \lambda(du_i,dy_i). \nonumber
\end{align}
This implies that $\gamma_i^*\otimes\mu_i = \lambda$, and so, $\gamma_i^*\otimes\mu_i(\sY_i \times \sU_i) = 1$. Hence, $\gamma_i^* \in \P(\sU_i)$ $\mu_i$-a.e. Thus, $\underline{\gamma}^*$ is an optimal policy.

The existence of deterministic optimal policy can proved as in the proof of \cite[Theorem 3.8]{GuYuBaLa15} and hence, we omit the details.
\end{proof}

\section{Existence of the Optimal Strategy for Dynamic Team Problems}\label{sec3}

The existence of team-optimal solutions for dynamic team problems can be established by a static reduction method. To this end, we first review the equivalence between dynamic teams and their static reduction (this is called {\it the equivalent model} \cite{wit88}\index{Witsenhausen's Equivalent Model}). Consider a dynamic team setting where there are $N$ decision epochs, and Agent~$i$ observes $y^i \sim W_i(\,\cdot\,|x,{\bf u}^{^{[1:i-1]}})$, and the decisions are generated as $u^i=\gamma^i(y^i)$. The resulting cost under a given team policy $\underline{\gamma}$ is
\[J(\underline{\gamma}) = E[c(x,{\bf y},{\bf u})].\]
This dynamic team can be converted to a static team provided that the following absolute continuity condition holds.

\begin{assumption}\label{as3}
For every $i=1,\ldots,N$, there exists a function $f_i: \sX\times\sU^{^{[1:i-1]}}\times\sY^i \rightarrow [0,\infty)$ and a probability measure $\nu_i$ on $\sY^i$ such that for all $S \in {\cal Y}^i$ we have
\begin{align}
W_i(S | x,{\bf u}^{^{[1:i-1]}}) = \int_{S} f_i(x,{\bf u}^{^{[1:i-1]}},y^i) \nu_i(dy^i). \nonumber
\end{align}
\end{assumption}

Therefore, for a fixed choice of $\underline{\gamma}$, the joint distribution of $(x,{\bf y})$ is given by
\[\rP(dx,d{\bf y}) = \rP(dx) \prod_{i=1}^N f_i(x,{\bf u}^{^{[1:i-1]}},y^i) \nu_i(dy^i),\]
where $\bu^{^{[1:i-1]}} = \bigl(\gamma^1(y^1),\ldots,\gamma^{i-1}(y^{i-1})\bigr)$. The cost function $J(\underline{\gamma})$ can then be written as
\begin{align}
J(\underline{\gamma}) &= \int_{\sX\times \sY} c(x,{\bf y},{\bf u}) \rP(dx) \prod_{i=1}^N f_i(x,{\bf u}^{^{[1:i-1]}},y^i) \nu_i(dy^i) \nonumber \\
&= \int_{\sX\times \sY} \tilde{c}(x,{\bf y},{\bf u}) \widetilde{\rP}(dx,d{\bf y}), \nonumber
\end{align}
where $\tilde{c}(x,{\bf y},{\bf u}) \coloneqq c(x,{\bf y},{\bf u})\prod_{i=1}^N f_i(x,{\bf u}^{^{[1:i-1]}},y^i)$ and
$\widetilde{\rP}(dx,d{\bf y}) \coloneqq \rP(dx)\prod_{i=1}^N \nu_i(dy^i)$. The observations now can be regarded as independent, and by incorporating the $f_i$ terms into $c$, we can obtain an equivalent {\it static team} problem. Hence, the essential step is to appropriately adjust the probability space and the cost function.

\subsection{Existence of Team-Optimal Policies}

In this section, using Theorem~\ref{main2} and the static reduction of the dynamic team problems, we prove the existence of an optimal policy.
Similar to the static case, we enlarge the space of policies such that agents can apply randomized policies as well.

Analogous to static case, we define, for any $L>0$, the following sets
\begin{align}
{\bf \Gamma}_L &\coloneqq \biggl\{\underline{\gamma} \in {\bf \Gamma}: J(\underline{\gamma}) < J^* + L  \biggr\} \nonumber \\
\intertext{and}
S_L &\coloneqq \biggl\{\lambda \in \P(\sX \times \sY \times \sU): \nonumber \\
&\phantom{xx}\lambda(dx,d\by,d\bu) = \rP(dx) \prod_{i=1}^N \gamma_i(du_i|y_i) \nu_i(dy_i), \underline{\gamma} \in {\bf \Gamma}_L \biggr\}. \nonumber
\end{align}
In addition, for all $j$, we define
\begin{align}
&S_{L,j} \coloneqq \biggl\{\lambda \in \P(\sX \times \sY \times \sU): \lambda(dx,d\by,d\bu) =  \rP(dx)  \nonumber \\
&\phantom{xx} \prod_{i=1}^j \gamma_i(du_i|y_i) \nu_i(dy_i)  \prod_{i=j+1}^N \gamma_i(du_i|y_i) W_i(dy_i|x,\bu^{^{[1:i-1]}}), \nonumber \\
&\phantom{xxxxxxxxxxxxxxxxxxxxxxxxxxxxxxxxxxxxxx}\underline{\gamma} \in {\bf \Gamma}_L \biggr\}. \nonumber
\end{align}
For each $i=1,\ldots,N$, define $S_L^i \coloneqq \Proj_{\sY_i \times \sU_i}(S_L) = \Proj_{\sY_i \times \sU_i}(S_{L,i})$. Then, we impose the following assumptions.

\begin{assumption}\label{as4}
Suppose that Assumption~\ref{as1}-(a),(b) and Assumption~\ref{as3} hold. In addition, we assume that
\begin{itemize}
\item [(a)] For all $i$, $f_i$ is lower semi-continuous.
\item [(b)] For some $L>0$, $S_L^i$ is tight for $i=1,\ldots,N$.
\end{itemize}
\end{assumption}

Before we continue with the main theorem of this section, we will give several conditions that imply Assumption~\ref{as4}-(b).

\begin{theorem}\label{tight2}
Suppose \emph{either }of the following conditions hold:
\begin{itemize}
\item [(i)] $\sU_i$ is compact for all $i$.
\item [(ii)] For non-compact case, we assume
\begin{itemize}
\item [(a)] The cost function $c$ is in class $\mathrm{IC}(\sX \times \sY^{^{[1:j]}} \times \sU^{^{[1:j-1]}},\sU_j)$, for all $j$.
\item [(b)] For all $j$, $f_j > 0$.
\end{itemize}
\end{itemize}
Then, Assumption~\ref{as4}-(b) holds.
\end{theorem}
\begin{proof}
(i): Note that the marginal on $\sY_i$ of any measure in $S_L^i$ is $\nu_i$. Since $\nu_i$ is tight by Proposition~\ref{aux2} and $\P(\sU_i)$ is tight by the compactness of $\sU_i$,  $S_L^i$ is also tight by Proposition~\ref{aux3}.\\
\noindent(ii):  For each $j$, we define
\begin{align}
\tilde{c}_j(x,\by,\bu) \coloneqq c(x,\by,\bu) \prod_{i=1}^j f_i(x,\bu^{^{[1:i-1]}},y_i). \nonumber
\end{align}
Since, for all $i$, $f_i$ is lower semi-continuous and strictly greater than $0$, for any compact set $K \subset \sX \times \sY^{^{[1:j]}} \times \sU^{^{[1:j-1]}}$, we have $\inf_{K} \prod_i^j f_i(x,\bu^{^{[1:i-1]}},y_i) > 0$. This implies that $\tilde{c}_j$ is also in class $\mathrm{IC}(\sX \times \sY^{^{[1:j]}} \times \sU^{^{[1:j-1]}},\sU_j)$. Then, by Theorem~\ref{aux5}, one can inductively prove that $\Proj_{\sX \times \sY^{^{[1:j]}} \times \sU^{^{[1:j]}}}(S_L)$ is tight for all $j$. Indeed, let $j=1$. Then $\tilde{c}_1$ is in $\mathrm{IC}(\sX \times \sY_{1},\sU_1)$ and
\begin{align}
&S_{L,1} \subset \biggl\{ \lambda \in \P(\sX \times \sY \times \sU): \mathrm{Proj}_{\sX \times \sY}(\lambda)(dx,dy_1) \nonumber \\
&\phantom{xxxxxxxx}= \rP(dx) \nu_1(dy_1) \text{ and } \int \tilde{c}_1 d\lambda \leq J^*+L \biggr\}. \nonumber
\end{align}
But since $\rP(dx) \nu_1(dy_1)$ is tight, by Theorem~\ref{aux5} $\Proj_{\sX \times \sY_{1} \times \sU_1}(S_{L,1}) = \Proj_{\sX \times \sY_{1} \times \sU_1}(S_L)$ is also tight. Suppose the assertion is true for $j$ and consider $j+1$. Since $\Proj_{\sY_{j+1}}(S_{L}) = \nu_{j+1}$ is tight by Proposition~\ref{aux2}, by Proposition~\ref{aux3} $\Proj_{\sX \times \sY^{^{[1:j+1]}} \times \sU^{^{[1:j]}}}(S_L)$ is also tight as $\Proj_{\sX \times \sY^{^{[1:j]}} \times \sU^{^{[1:j]}}}(S_L)$ is tight by the induction hypothesis.
Note that $\tilde{c}_{j+1}$ is in $\mathrm{IC}(\sX \times \sY^{^{[1:j+1]}} \times \sU^{^{[1:j]}}, \sU_{j+1})$ and
\begin{align}
&S_{L,j+1} \subset \biggl\{ \lambda \in \P(\sX \times \sY \times \sU): \mathrm{Proj}_{\sX \times \sY^{^{[1:j+1]}} \times \sU^{^{[1:j]}}}(\lambda) \nonumber \\
&\phantom{x} \in \mathrm{Proj}_{\sX \times \sY^{^{[1:j+1]}} \times \sU^{^{[1:j]}}}(S_L) \text{ and } \int \tilde{c}_{j+1} d\lambda \leq J^*+L \biggr\}. \nonumber
\end{align}
Since $\mathrm{Proj}_{\sX \times \sY^{^{[1:j+1]}} \times \sU^{^{[1:j]}}}(S_L)$ is tight,
\begin{align}
\mathrm{Proj}_{\sX \times \sY^{^{[1:j+1]}} \times \sU^{^{[1:j+1]}}}(S_{L,j+1}) = \mathrm{Proj}_{\sX \times \sY^{^{[1:j+1]}} \times \sU^{^{[1:j+1]}}}(S_L) \nonumber
\end{align}
is also tight by Theorem~\ref{aux5}. This completes the proof of assertion. But this result implies that $S_L^j$ is tight for all $j$.
\end{proof}

\begin{theorem}\label{main3}
Suppose Assumptions~\ref{as4} holds. Then, the static reduction of the dynamic team model satisfies Assumptions~\ref{as1} and \ref{as2}. Therefore, by Theorem~\ref{main2}, there exists an optimal strategy for dynamic team problem.
\end{theorem}

\section{Witsenhausen's Counterexample}\label{witsen}

In Witsenhausen's celebrated counterexample \cite{Wit68}, depicted in Fig.~\ref{fig1}, there are two decision makers: Agent~$1$ observes a zero mean
and unit variance Gaussian random variable $y_1$ and decides its strategy $u_1$. Agent~$2$ observes $y_2 \coloneqq u_1 + v$, where $v$ is standard (zero
mean and unit variance) Gaussian noise independent of $y_1$, and decides its strategy $u_2$.

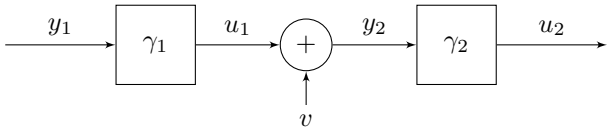
\begin{figure}[h]
\centering
\tikzstyle{int}=[draw, fill=white!20, minimum size=3em]
\tikzstyle{init} = [pin edge={to-,thin,black}]
\tikzstyle{sum} = [draw, circle]
\begin{tikzpicture}[node distance=2cm,auto,>=latex']
    \node [int] (a) {$\gamma_1$};
    \node (b) [left of=a,node distance=2cm, coordinate] {a};
    \node [sum] (d) [right of=a] {$+$};
    \node [int] (c) [right of=d] {$\gamma_2$};
    \node [coordinate] (end) [right of=c, node distance=2cm]{};
    \node (e) [below of=d, node distance=1cm] {$v$};
    \path[->] (e) edge node {} (d);
    \path[->] (b) edge node {$y_1$} (a);
    \path[->] (a) edge node {$u_1$} (d);
    \path[->] (d) edge node {$y_2$} (c);
    \path[->] (c) edge node {$u_2$} (end) ;
\end{tikzpicture}
\caption{Witsenhausen's counterexample.}
\label{fig1}
\end{figure}

The cost function of the team is given by
\begin{align}
c(y_1,u_1,u_2) = l (u_1 - y_1)^2 + (u_2 - u_1)^2, \nonumber
\end{align}
where $l \in \R_{+}$. In this decentralized system, the state of the nature $x$ can be regarded as a
degenerate (constant) random variable. Let $g(y) \coloneqq \frac{1}{\sqrt{2\pi}}\exp{\{-y^2/2\}}$. Then we have
\begin{align}
\rP(y_2 \in S|u_1) = \int_{S} g(y_2-u_1) m(dy_2), \nonumber
\end{align}
where $m$ denotes the Lebesgue measure. Let
\begin{align}
f(u_1,y_2) \coloneqq \exp{\bigl\{-\frac{(u_1)^2-2y_2u_1}{2}\bigr\}} \label{eq20}
\end{align}
so that $g(y_2-u_1) =f(u_1,y_2)\frac{1}{\sqrt{2\pi}}\exp{\{-(y_2)^2/2\}}$. The static reduction proceeds as follows: for any policy $\underline{\gamma}$, we have
\begin{align}
J(\underline{\gamma}) &= \int c(y_1,u_1,u_2) \rP(dy_2|u_1) \delta_{\gamma_1(y_1)}(du_1) \rP_{\sg}(dy_1) \nonumber \\
&= \int c(y_1,u_1,u_2) f(u_1,y_2) \rP_{\sg}(dy_2) \rP_{\sg}(dy_1), \nonumber
\end{align}
where $\rP_{\sg}$ denotes the standard Gaussian distribution. Hence, by defining $\tilde{c}(y_1,y_2,u_1,u_2) = c(y_1,u_1,u_2) f(u_1,y_2)$ and $\widetilde{\rP}(dy_1,dy_2) = \rP_{\sg}(dy_1) \rP_{\sg}(dy_2)$, we can write $J(\underline{\gamma})$ as
\begin{align}
J(\underline{\gamma}) = \int \tilde{c}(y_1,y_2,u_1,u_2) \widetilde{\rP}(dy_1,dy_2). \label{eq4}
\end{align}
Therefore, in the static reduction of Witsenhausen's counterexample, the agents observe independent zero mean and unit variance Gaussian random variables.

To tackle the existence problem for Witsenhausen's counterexample we show that the conditions in Theorem~\ref{main3} hold.

\begin{theorem}\label{witsen_exist}
Witsenhausen's counterexample satisfies conditions in Theorem~\ref{main3}. Hence, there exists an optimal policy.
\end{theorem}

\begin{proof}
Assumption~\ref{as1}-(a),(b), Assumption~\ref{as3}, and Assumption~\ref{as4}-(a) clearly hold. To prove Assumption~\ref{as4}-(b), we use Theorem~\ref{tight2}. Indeed, it is clear that
$c$ is both in $\mathrm{IC}(\sY_1,\sU_1)$ and in $\mathrm{IC}(\sY_1 \times Y_2 \times \sU_1,\sU_2)$. Hence, by Theorem~~\ref{tight2}, Assumption~\ref{as4}-(b) holds.
\end{proof}

\section{The Gaussian Relay Channel Problem}\label{gaussrelay}

An important dynamic team problem which has attracted interest is the Gaussian relay channel problem \cite{LipsaMartins,zaidi2013optimal} depicted in Fig.~\ref{fig2}. Here, Agent~$1$ observes a noisy version of the state $x$ which has Gaussian distribution with zero mean and variance $\sigma_{x}^2$; that is, $y_1 \coloneqq x + v_0$ where $v_0$ is a zero mean and
variance $\sigma_0^2$ Gaussian noise independent of $x$. Agent~$1$ decides its strategy $u_1$ based on $y_1$. For $i=2,\ldots,N$, Agent~$i$ receives $y_i \coloneqq u_{i-1} + v_{i-1}$ (a noisy version of the decision $u_{i-1}$ of Agent~$i-1$), where $v_{i-1}$ is  zero mean and variance $\sigma_{i-1}^2$ Gaussian noise independent of $\{x,v_{1},\ldots,v_{i-2},v_{i},\ldots,v_{N-1}\}$, and decides its strategy $u_{i}$.

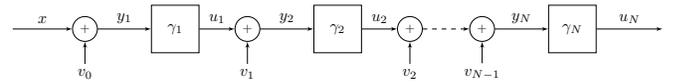
\begin{figure}[h]
\centering
\tikzstyle{int}=[draw, fill=white!20, minimum size=3em]
\tikzstyle{init} = [pin edge={to-,thin,black}]
\tikzstyle{sum} = [draw, circle, scale=0.8]
\scalebox{0.6}{
\begin{tikzpicture}[node distance=2cm,auto,>=latex']
    \node (x) [coordinate, node distance=2cm] {};
    \node [sum] (s0) [right of=x] {$+$};
    \node (v0) [below of=s0, node distance=1cm] {$v_0$};
    \node [int] (c1) [right of=s0] {$\gamma_1$};
    \node [sum] (s1) [right of=c1] {$+$};
    \node (v1) [below of=s1, node distance=1cm] {$v_1$};
    \node [int] (c2) [right of=s1] {$\gamma_2$};
    \node [sum] (s2) [right of=c2] {$+$};
    \node (v2) [below of=s2, node distance=1cm] {$v_2$};
    \node [sum] (s3) [right of=s2] {$+$};
    \node (v3) [below of=s3, node distance=1cm] {$v_{N-1}$};
    \node [int] (c3) [right of=s3] {$\gamma_{N}$};
    \node [coordinate] (end) [right of=c3, node distance=2cm] {};
    \path[->] (x) edge node {$x$} (s0);
    \path[->] (v0) edge node {} (s0);
    \path[->] (s0) edge node {$y_1$} (c1);
    \path[->] (c1) edge node {$u_1$} (s1);
    \path[->] (v1) edge node {} (s1);
    \path[->] (s1) edge node {$y_2$} (c2);
    \path[->] (c2) edge node {$u_2$} (s2);
    \path[->] (v2) edge node {} (s2);
    \path[->,dashed,->] (s2) edge node {} (s3);
    \path[->] (v3) edge node {} (s3);
    \path[->] (s3) edge node {$y_N$} (c3);
    \path[->] (c3) edge node {$u_N$} (end) ;
\end{tikzpicture}}
\caption{Gaussian relay channel.}
\label{fig2}
\end{figure}

The cost function of the team is given by
\begin{align}
c(x,\bu) \coloneqq \bigl( u_N - x \bigr)^2 + \sum_{i=1}^{N-1} l_i \bigl(u_i\bigr)^2, \nonumber
\end{align}
where $l_i \in \R_+$ for all $i=1,\ldots,N-1$. To ease the notation, we simply take $\sigma_x = \sigma_0 = \sigma_1 = \ldots = \sigma_{N-1} =1$.  Recall that
$g(y) \coloneqq \frac{1}{\sqrt{2\pi}} \exp{\{-y^2/2\}}$. Then we have
\begin{align}
\rP(y_1 \in S|x) &= \int_{S} g(y_1-x) m(dy_1)  \nonumber \\
\rP(y_i \in S|u_{i-1}) &= \int_{S} g(y_i-u_{i-1}) m(dy_{i}), \text{  for } i=2,\ldots,N. \nonumber
\end{align}
Recall also that $g(y-u) =f(u,y)\frac{1}{\sqrt{2\pi}}\exp{\{-(y)^2/2\}}$ , where $f(u,y)$ is defined in (\ref{eq20}).
Then, for any policy $\underline{\gamma}$, we have
\begin{align}
J(\underline{\gamma}) &= \int_{\sX \times \sY} c(x,\bu) \rP(dx,d\by) \nonumber \\
&= \int_{\sX\times\sY} c(x,\bu) \biggr[ f(x,y_1) \prod_{i=2}^N f(u_{i-1},y_i) \biggr] \text{ } \rP_{\sg}^{N+1}(dx,d\by), \nonumber
\end{align}
where $\rP_{\sg}^{N+1}$ denotes the product of $N+1$ zero mean and unit variance Gaussian distributions. Therefore, in the static reduction of Gaussian relay channel, we have the components $\tilde{c}(x,\by,\bu) \coloneqq c(x,\bu) \bigr[ f(x,y_1) \prod_{i=2}^N f(u_{i-1},y_i) \bigr]$ and $\widetilde{\rP}(dx,d\by) = \rP_{\sg}^{N+1}(dx,d\by)$. Analogous to Witsenhausen's counterexample, the agents observe independent zero mean and unit variance Gaussian random variables.

\begin{theorem}\label{relay_exist}
The Gaussian relay channel problem satisfies the conditions in Theorem~\ref{main3}. Hence there exists an optimal policy.
\end{theorem}

\begin{proof}
Assumption~\ref{as1}-(a),(b), Assumption~\ref{as3}, and Assumption~\ref{as4}-(a) clearly hold. To prove Assumption~\ref{as4}-(b), we use Theorem~\ref{tight2}. Indeed, it is clear that
$c$ is in $\mathrm{IC}(\sX \times \sY^{^{[1:j]}} \times \sU^{^{[1:j-1]}},\sU_j)$, for all $j$. Hence, by Theorem~~\ref{tight2}, Assumption~\ref{as4}-(b) holds.
\end{proof}

\section{Conclusion}\label{conc}

Existence of team-optimal policies for both static and dynamic team problems was considered. Under mild technical conditions, we first showed the existence of an optimal policy for static teams. Using this result,  analogous existence results were also established for the dynamic teams via Witsenhausen's static reduction method. Finally, we apply our findings to well-known counterexample of Witsenhausen and Gaussian relay channel.

\bibliographystyle{IEEEtran}
\bibliography{references,SerdarBibliography}


\end{document}